\RequirePackage{rotating}
\documentclass[12pt,a4paper]{amsart}
\usepackage{amssymb}
\usepackage{tikz-cd}
\usepackage{chngcntr, csquotes}
\usepackage[margin=1in]{geometry}

\usepackage{accents}
\usepackage{todonotes}
\usepackage[shortlabels]{enumitem}

\usepackage{hyperref}

\hypersetup{
  colorlinks=true,breaklinks,
  linktoc=section,
  linkcolor=blue,
  citecolor=blue,
  urlcolor=blue,
}

\usepackage{multirow}

\usepackage{lineno}

\usepackage{float,rotating}
\setuptodonotes{inline}

\usepackage{array}
\newcolumntype{L}[1]{>{\raggedright\let\newline\\\arraybackslash\hspace{0pt}}m{#1}}
\newcolumntype{C}[1]{>{\centering\let\newline\\\arraybackslash\hspace{0pt}}m{#1}}
\newcolumntype{R}[1]{>{\raggedleft\let\newline\\\arraybackslash\hspace{0pt}}m{#1}}

\newcolumntype{P}[1]{>{\centering\arraybackslash}p{#1}}
\newcolumntype{Q}[1]{>{\raggedleft\arraybackslash}p{#1}}

\usepackage{extdash}

\newcounter{proofcount}
\AtBeginEnvironment{proof}{\stepcounter{proofcount}}
\newtheorem{claim}{Claim}
\makeatletter                  
\@addtoreset{claim}{proofcount}
\makeatother                   

\theoremstyle{remark}
\makeatletter
\newtheorem*{cproof/}{Proof of Claim \rev@cproofmark}

\newenvironment{cproof}[1][\@nil]
  {\def\@tmp{#1}%
   \ifx\@tmp\@nnil
       \def\rev@cproofmark{\theclaim}
    \else
       \let\rev@cproofmark\@tmp
    \fi
   \pushQED{\qed}\begin{cproof/}}
  {\popQED\end{cproof/}}
\makeatother

\usepackage{cleveref}

\theoremstyle{theorem}
\newtheorem{thm}{Theorem}[section]
\newtheorem{lem}[thm]{Lemma}

\newtheorem{dfn}[thm]{Definition}

\theoremstyle{definition}
\newtheorem{question}[thm]{Question}

\newtheorem*{rem}{Remark}

\crefname{thm}{theorem}{theorems}
\Crefname{thm}{Theorem}{Theorems}
\crefname{lem}{lemma}{lemmas}
\Crefname{lem}{Lemma}{Lemmas}
\crefname{dfn}{definition}{definitions}
\Crefname{dfn}{Definition}{Definitions}
\crefname{claim}{claim}{claims}
\Crefname{claim}{Claim}{Claims}

\numberwithin{equation}{section}

\newcommand{\e}[1]{\overline{#1}}
\newcommand{\sstar}{\Sigma^*}
\newcommand{\sep}{\mathord\circ}
\newcommand{\spp}{\mathord\bullet}
\newcommand{\stpp}{\mathord\star}
\newcommand{\fin}{\mathord\wr}
\newcommand{\gapp}{\, \sep\spp \,}
\newcommand{\stopp}{\, \sep\stpp}

\newcommand{\upgap}{\, | \,}

\newcommand{\rowgap}[1]{\multicolumn{#1}{c}{\vspace{-10pt}}}

\newcommand{\sub}{\mathsf{sub}}

\newcommand{\NP}{\mathbf{NP}}
\renewcommand{\a}{\alpha}
\renewcommand{\b}{\beta}
\newcommand{\ra}{\rightarrow}
\newcommand{\Ra}{\Rightarrow}
\newcommand{\emptyword}{\varepsilon}

\newcommand{\set}[2]{\{ \, #1 \colon #2 \,\}}

\newcommand{\define}[1]{{\normalfont{\textbf{#1}}}}

\begin{document}

\title[Non-Terminal Complexity of Simple Semi-Conditional Grammars]{Non-Terminal Complexity\\ of Simple Semi-Conditional Grammars}

\author[H.~Fernau]{Henning Fernau}
\address{FB IV -- Informatikwissenschaften Theoretische Informatik Campus II / Geozentrum Geb\"{a}ude H, Universit\"{a}t Trier 54286 Trier, Germany}
\email{fernau@uni-trier.de}

\author[S.~Jain]{Sanjay Jain}
\address{School of Computing, National University of Singapore, 13 Computing Drive, Singapore 117417, Republic of Singapore}
\email{sanjayjain@nus.edu.sg}

\author[L.~Richter]{Linus Richter}
\address{Flinders University, College of Science and Engineering, Level 3, Tonsley Building 1, 1284 South Road, Tonsley SA 5042, Australia}
\email{linus.richter@flinders.edu.au}
\urladdr{https://linus-richter.github.io}

\author[F.~Stephan]{Frank Stephan}
\address{School of Computing, National University of Singapore, 13 Computing Drive, Singapore 117417 and Department of Mathematics, National University of Singapore, 10 Lower Kent Ridge Road, Singapore 119076}
\email{fstephan@nus.edu.sg}
\urladdr{https://www.comp.nus.edu.sg/~fstephan/}

\author[D.~Turetsky]{Dan Turetsky}
\address{School of Mathematics and Statistics Victoria University of Wellington Cotton Building, Gate 7, Kelburn Parade Wellington, New Zealand}
\email{dan.turetsky@vuw.ac.nz}

\date{\today}

\thanks{Sanjay Jain was supported by NUS Provost chair grant E-252-00-0021-09. Linus Richter was fully supported by Singapore Ministry of Education grant MOE-000538-01. Additionally, Sanjay Jain and Frank Stephan were partially supported by Singapore Ministry of Education grant MOE-000538-01. Dan Turetsky was supported by the Royal Society of New Zealand.}


\begin{abstract}
	We study the complexity of simple semi-conditional grammars (SSCGs) in terms of the number of their terminals and non-terminals. We show that SSCGs with three non-terminals can generate all RE languages; two non-terminals suffice for linear languages; and one suffices for unary regular languages. We also determine both upper bounds and fundamental limitations of SSCGs: while certain one-non-terminal SSCGs can generate non-context-free languages, every language generated by a one-non-terminal SSCG is in non-deterministic linear space. We also prove that there exists a regular language of alphabet size~$\ge 3$ which cannot be generated by any one-non-terminal SSCG. Finally, we prove that membership testing for SSCGs of two non-terminals is already $\NP$-hard.
\end{abstract}

\maketitle

\section{Introduction}\label{sec:intro}

\subsection{Background}
We study the descriptive complexity of certain regulated rewriting grammar formalisms which have been previously proven to be computationally complete. More precisely, we investigate the \emph{non-terminal complexity} of such grammars. Our motivating question is:

\begin{quote}
What is the smallest number of non-terminals sufficient to describe each recursively enumerable (RE) language?
\end{quote}

For the classical rewriting mechanisms---such as graph-controlled, programmed, or matrix grammars (all with appearance checking)---this question was studied in a sequence of papers, finally bringing the number of required non-terminals down to 2, 3, and~3, respectively~\cite{Feretal07}. There, it was also shown that one non-terminal is insufficient to produce all RE languages, but sufficient to obtain all linear languages with graph control. Conversely, it was also remarked at the end of~\cite{Feretal07} that Example~4.1.1 in~\cite{DasPau89} shows that this dramatic collapse of the non-terminal hierarchy does not occur with random context grammars with appearance checking; in contrast, for these grammars, the family of languages generable with~$k$ non-terminals is a proper subclass of the family of languages generable with~$k+1$ non-terminals. Hence, it is interesting to study related grammatical mechanisms. This is why we study so-called \emph{simple semi-conditional grammars} in this paper.

\medskip

Generalising conditional grammars due to Kelemen in~1984~\cite{Kel84}, P\u{a}un introduced \emph{semi-conditional grammars} (SCGs) one year later~\cite{Pau85}. Here, each rule is associated with two strings called the \emph{permitting} and \emph{forbidden string}. A rule can be applied to a sentential form~$w$ only if~$w$ contains the permitting string (i.e., positive context) and does not contain the forbidden string (i.e., negative context) as a subword.
A semi-conditional grammar is termed \emph{simple} (denoted SSCG) if for each rule at most either the permitting string or the forbidden string is present~\cite{MedGol94}. Notice that SSCGs generalise Kelemen's conditional grammars, too.
Quite some research focused on the so-called \emph{degree} of (simple) semi-conditional grammars, which is the pair formed by upper-bounds on the length of the permitting and forbidden strings. Based on these bounds, further bounds on other descriptive complexity measures were investigated~\cite{FerKOR2021a}. For instance, the number of conditional rules have been of interest in~\cite{MedSve2002}.

\medskip
 
In this paper, our main focus is the number of non-terminals of these grammars. This measure of descriptive complexity has been investigated before with semi-conditional grammars, but only in passing. In~\cite{FerKupOla2018}, it was shown that five non-terminals suffice in the case of SCGs of degree~$(4,1)$ to describe every RE language, while in \cite{FerKOR2021a}, it was proven that six non-terminals are enough for this purpose for SSCGs.
What can be achieved if only one, two, or three non-terminals are allowed, disregarding other measures of descriptive complexity?

\subsection{Our Theorems and the Structure of This Paper}

We show that every RE language can be generated by an SSCG with three non-terminals (\Cref{thm-relang}). Furthermore, it is shown that every linear language can be generated by an SSCG using two non-terminals (\Cref{thm-linear}), and every unary regular (i.e.~unary context-free) language can be generated using one non-terminal (\Cref{thm-unaryregular}). However, there are regular languages over alphabets of size~3 which cannot be generated by any SSCG with only one non-terminal (\Cref{thm-regular-needtwo}). It is open at present if every regular language over a binary alphabet can be generated by an SSCG with two non-terminals.

We also show that there are SSCGs with one non-terminal which generate a non-context free language (Theorem~\ref{thm-oneNT-nonCF}). However, a language generated by any SSCG with one non-terminal is in  non-deterministic linear space (\Cref{thm-oneNT-linearspace}). Finally, we give a complexity characterisation for membership testing: for SSCGs with only 2 non-terminals, membership testing is already $\NP$-hard (\Cref{thm-nphardness}).

\bigskip

A summary of our results can be found in \Cref{table:results}, where we refer to the Chomsky Hierarchy explicitly. We refer the reader to the Complexity Zoo \cite{zoo} for further details on complexity classes and their relationships. For a survey on regulated rewriting, we refer to the textbook of Dassow and P\u{a}un~\cite{DasPau89}.

\begin{sidewaystable}[htbp]
\setlength{\tabcolsep}{12pt}
\setlength{\extrarowheight}{3pt}
\centering
\begin{tabular}{ccccc}
\multirow{2}{*}{Shortform} & \multirow{2}{*}{Name} & \multirow{2}{*}{Condition on SSCG} & {Condition on} & \multirow{2}{*}{Class}  \\
   & & & {Usual Grammar} & \\
   \hline
   \rowgap{5}\\
   \hline
\multirow{2}{*}{unary REG} & Regular with & SBS/EQT: $\Delta = \{S\}$ & \multirow{2}{*}{see below} & EQT: Unary \\
          & unary alphabet & {\small (\Cref{thm-unaryregular})} & & Context-Free \\
\hline
\multirow{2}*{REG, CH3} & \multirow{2}*{Regular} & SBS: $\Delta = \{S,T\}$ & \multirow{2}*{$A \ra aB | b$} &
\multirow{2}*{EQT: SPACE$(O(1))$} \\
& & {\small (\Cref{thm-regular-needtwo})} &&\\
\hline
\multirow{3}{*}{LIN} & \multirow{3}{*}{Linear}  & SBS: $\Delta = \{S,T\}$ and & \multirow{3}{*}{$A \ra aB|Ba|a$} &
\multirow{3}{*}{SBS: NLOGSPACE} \\
 & & $A \ra vaw$ & & \\
 & & {\small (\Cref{thm-linear})} &\\
\hline
\multirow{2}{*}{CH2, CF} & \multirow{2}{*}{Context-Free}  & SBS: $\Delta = \{S,T,U\}$ & \multirow{2}{*}{$A \ra BC|a$} &
SBS: POLYLOGSPACE \\
& & {\small (Theorems~\ref{thm-regular-needtwo} +~\ref{thm-relang})}  & & SBS: P, Nick's Class \\
\hline
\multirow{2}{*}{CS, CH1} & Context- & \multirow{2}{*}{EQT: $\emptyword$-free} & $v \ra w$ & 
\multirow{2}{*}{EQT: NLINSPACE} \\
  & Sensitive & &  with $|v| \leq |w|$ & \\
\hline
\multirow{2}{*}{---} & decidable, & \multirow{2}{*}{---} & \multirow{2}{*}{---} & characteristic func- \\
    & recursive  &     &     & tion total recursive \\
\hline
\multirow{2}{*}{RE}  & recursively & EQT: $\Delta = \{S,T,U\}$ & \multirow{2}{*}{no restriction} & ranges of enumer-\\
    & enumerable & {\small (\Cref{thm-relang})} & & ation procedures \\
\end{tabular}
\vspace*{1em}
\caption{EQT means ``equal to,'' SBS means ``is contained in,'' SPS means ``strictly contains.'' For SSCGs, $\Delta$ denotes the set of distinct non-terminals.}
\label{table:results}
\end{sidewaystable}

\section{Formal Definitions and Conventions}

We assume a basic understanding of formal language theory as provided in standard references (e.g.~\cite{HopUll79,Sal73}). In particular, we assume familiarity with the language classes of the Chomsky hierarchy: regular, (linear,) context-free, context-sensitive, (recursive,) and recursively enumerable, listed in ascending order of complexity. In parentheses, we have added well-known intermediate classes.

\medskip

All alphabets in this paper are finite.

\medskip

We now define (simple) semi-conditional grammars. We fix some notation: if $w$ is a word formed with letters from the alphabet~$\Sigma$, written as $w\in\sstar$, then $y\in \Sigma^*$ is a \define{subword} (sometimes also called a factor or infix) of~$w$ if there are words $x,z\in\Sigma^*$ such that $w=xyz$. The set of subwords of~$w$ is written as $\sub(w)$. As is customary,~$\emptyword$ denotes the empty word, and  we define~$\Sigma^+ = \Sigma^* - \{ \emptyword \}$.

\begin{dfn}
A \define{semi-conditional grammar} is a quadruple $G = (V_N,V_T,P,S)$, where 
\begin{itemize}
\item $V_N$ is the non-terminal alphabet,
\item $V_T$ is the terminal alphabet, $V_T \cap V_N=\emptyset$,
\item $V=V_N\cup V_T$ is the total alphabet,
\item $S \in V_N$ is the starting symbol, 
\item $P$ is a finite set of (labeled) productions of the form $\ell \colon (A \to x, \a, \b)$ with
\begin{itemize}
\item $A \in V_N$
\item $x \in V^*$
\item $\a, \b \in V^+\cup \{0\}$, where $ 0 \not\in V$ is a special symbol, intuitively meaning that the condition is missing
\item $\ell$ is the label of the production
\end{itemize}
\end{itemize}
\end{dfn}

The production~$\ell \colon (A \to x, \a, \b)$ can be applied to a word~$u \in V^*V_NV^*$ if and only if
\begin{center}
	$A\in \sub(u)$ and ($\a \in \sub(u)$ or $\a=0$) and ($\b \not \in \sub(u)$ or $\b=0$).
\end{center}
Under these conditions, $u = u_1Au_2$ will be transformed into $v = u_1xu_2$, which is denoted by $u\Ra_\ell v$. In such a case, we call $\a$ \define{permitting} and $\b$ \define{forbidden}.
When there is no risk of confusion, we avoid mentioning the label and simply write $u \Ra v$. The \define{language} of $G$ is defined as
\[
	L(G) = \set{w \in V_T^*}{S \Ra^* w }
\]
where $\Ra^*$ denotes the reflexive and transitive closure of $\Ra$.

\begin{dfn}\label[dfn]{dfn:sscg}
A semi-conditional grammar $G$ is called \define{simple} (denoted by SSCG) if for every production rule $\ell \colon (A \to x, \a, \b) \in P$ we have~$\a = 0$ or~$\b = 0$.
\end{dfn} 

A production rule $\ell \colon (A \to x, \a, \b) $ in an SSCG is called a \define{presence rule} if~$\b=0$, and an \define{absence rule} if~$\a=0$. For ease of notation, for presence/absence rules, we call the corresponding strings $\a$ and $\b$ \define{presence/absence strings}. A rule is \define{applicable} to some sentential form if the corresponding condition of said rule is satisfied by the sentential form.

\bigskip

We denote non-terminals by $S,T,U$ or $S,T$ or $S$, depending on the number of terminals required. Terminals are indexed by natural numbers; for the $k$-element language we hence use~$1,\ldots,k$, as well as~$0$ which, we recall, is a special symbol.

Further, we fix the following:
\begin{itemize}
\item Terminal symbols are denoted by letters from the beginning of the alphabet:
\[
	a,b,c,\ldots
\]
\item Words (terminals and non-terminals) are denoted by letters from the end of the alphabet:
\[
	u,v,w,x,y,z
\]
\item Non-terminals are denoted by capital letters from the beginning of the alphabet:
\[
	A,B,C,\ldots
\]
In simulated languages, we index the set of non-terminals by the natural numbers:
\[
	Z_1,Z_2,Z_3,\ldots
\]
In such cases, we decide that~$Z_0$ is the starting symbol.
\end{itemize}

\section{Two Non-Terminals Suffice for Linear Languages}
Our first result shows that any linear language, and thus any regular language, can be generated by an SSCG with two non-terminals.

\begin{thm}\label{thm-linear}
Linear languages can be generated by simple semi-conditional grammars with non-terminals~$S,T$ only.  
\end{thm}

\begin{proof}
Suppose the linear grammar~$G$ for a language $L$ has non-terminals~$Z_0,Z_1, Z_2, \ldots$, where~$Z_0$ is the starting symbol in~$G$. By an intermediate application of the Chomsky Normal Form Theorem, all rules in~$G$ may be assumed to be of the form~$Z_i \ra Z_j a$ or~$Z_i \ra b Z_j$ or $Z_i \ra c$, where $a,b,c \neq \emptyword$ are members of the alphabet.

Fix two non-terminals~$S,T$, where~$S$ denotes the starting symbol in our SSCG. The following table describes the SSCG for the language~$L$. By~\Cref{dfn:sscg}, it suffices to define presence and absence rules:

\begin{center}
\setlength{\extrarowheight}{3pt}
\begin{tabular}{c|c|c}
Rule & Condition on Subword & Condition on $G$\\
\hline
\rowgap{3}\\
\hline
$S \ra \emptyword$ & $T$ absent & $\emptyword \in L$\\
$S \ra a$ & $T$ absent & $Z_0 \ra a$\\
$S \ra aT^iST^i$ & $T$ absent & $Z_0 \ra a Z_i$\\
$S \ra T^iST^ia$ & $T$ absent & $Z_0 \ra Z_i a$\\
\hline
$S \ra aT^jST^j$ & $bT^iST^i$ present & $Z_i \ra a Z_j$\\
$S \ra aT^jST^j$ & $T^iST^ib$ present & $Z_i \ra a Z_j$\\
$S \ra T^jST^ja$ & $bT^iST^i$ present & $Z_i \ra Z_j a$\\
$S \ra T^jST^ja$ & $T^iST^ib$ present & $Z_i \ra Z_j a$\\
$S \ra a$ & $bT^iST^i$ present & $Z_i \ra a$\\
$S \ra a$ & $T^iST^ib$ present & $Z_i \ra a$\\
\hline
$T \ra \emptyword$ & $S$ absent & after termination\\
\end{tabular}
\end{center}

Observe that since the alphabet of~$G$ is finite, the table above is also finite. Hence, the induced SSCG is indeed well-defined. 

\subsection*{Verification} Let~$u \in L$ be given by a derivation sequence~$S \Ra^* u$. Suppose the derivation has steps~$u_1,\dots,u_k$. By definition of linear grammars, each~$u_i$ for~$i < k$ contains exactly one~$Z_{n_i}$, where~$u_1 = Z_{n_1} = Z_0$ and~$u_k = u$.

Let~$v_1,\dots,v_k$ be a parallel derivation operating on our SSCG, where~$v_1 = S$, and~$v_{i+1}$ is defined in the obvious fashion to replicate the operation in~$u_{i+1}$. Observe that, except possibly at the start, each copy of $T^i S T^i$ has a character from $\Sigma$ either to its left or to its right, but not both.
For each~$i$, let~$\e{u}_i,\e{v}_i \in \sstar$ denote the words obtained by deleting all instances of~$Z_j$ for all~$j$ in~$\e{u}_i$, and all instances of~$S$ and~$T$ in~$\e{v}_i$.
It is easily seen by induction that~$\e{u}_i = \e{v}_i$ for all~$i \le k$. The other direction is obtained analogously, mutatis mutandis.
\end{proof}

\Cref{thm-linear} is optimal: in~\Cref{thm-regular-needtwo} below we show that there exists a regular (and hence linear) language which cannot be generated by an SSCG with only one non-terminal.

\section{Regular Languages}

In this section, we consider the strength of SSCGs with one non-terminal for generating regular languages. First, we show that every unary regular language (and thus every unary context-free language) can be generated by an SSCG with one non-terminal. In~\Cref{thm-regular-needtwo}, we show that there are regular languages which cannot be generated by SSCGs having only one non-terminal.

\begin{thm}\label{thm-unaryregular}
Unary context-free languages can be recognised
by simple semi-conditional grammars with one non-terminal.   
\end{thm}
\begin{proof}
Consider any context-free language~$L$ over a unary alphabet~$\{ 1 \}$, and recall that such languages are always regular~\cite{GinsRice}. Further, note that, by Parikh's Theorem, there exists an~$h$ such that~$L=A \cup B \cdot \left( 1^h \right)^*$, where~$A$ and~$B$ are finite sets and~$ \emptyword \not \in B$. The following SSCG with non-terminal~$S$ recognises~$L$.

\begin{center}
\setlength{\extrarowheight}{3pt}
\begin{tabular}{rc|c|c}
& Rule & Condition on Subword & Condition on $L$\\
\hline
\rowgap{4}\\
\hline
(1) & $S \ra w$ & $1$ absent & $w \in A$\\
(2) & $S \ra wS$ & $1$ absent & $w \in B$\\
(3) & $S \ra 0^hS$ & $1$ present & ---\\
(4) & $S \ra \emptyword$ & $1$ present & ---\\
\end{tabular}
\end{center}

\subsection*{Verification} If~$w \in A \cup B$ then clearly our SSCG generates it through rules~(1), (2), and~(4). Otherwise,~$w = u 1^{(nh)}$ for some~$u \in B$ and some~$n$. Our SSCG obtains~$w$ then through rule~(2), followed by~$n$ applications of rule~(3), and one application of rule~(4). It is clear that these are the only words generated by our SSCG.
\end{proof}

We now show that there exists a regular language over three characters which cannot be obtained through a simple semi-conditional grammar with only one non-terminal.

\begin{thm} \label{thm-regular-needtwo}
Fix~$\Sigma = \{ a,b,c \}$ and let~$L = a^* \cup b^* \cup c^* $. There is no simple semi-conditional grammar with only one non-terminal~$S$ which generates~$L$.
\end{thm}

\begin{proof}
We will show that every SSCG of exactly one non-terminal which generates~$L$ also generates a string containing at least two distinct characters from~$\Sigma$, which, by definition of~$L$, proves the theorem.

In our proof, we carry out a case analysis based on the different productions. To limit the number of those---and hence the number of cases we need to consider---we note the following: 
\begin{itemize}
	\item Without loss of generality, assume that no side condition is ``presence of~$S$'' or ``absence of~$S$''. These can be omitted without loss of generality.
	\item If there is a rule~$S \ra \beta$ in absence of~$\alpha$, and there is a rule~$S \ra \beta$ in presence of~$\alpha'$, with~$\alpha'$ being a substring of~$\alpha$, then we can drop these rules and replace them by~$S \ra \beta$, without any side condition.
	\item Since~$\emptyword \in L$ and presence rules cannot eliminate all~$S$, there must either be an unconditional rule~$S \ra \emptyword$ or an absence rule~$S \ra \emptyword$.
\end{itemize}

We now move on to our detailed case analysis. This leads to a lot of short-lived notation, which we try to preempt by the following convention:
\begin{itemize}
	\item Numbered rules (e.g.~$R_{1,d}$,~$R_3$) are \define{global names}. In contrast, rules without subscripts (e.g.~$R$,~$R'$,) are \define{local names} and only refer to the case/subcase in which they appear.
\end{itemize}

We now commence the proof proper.

\medskip

Let~$G$ be an SSCG. Consider the following rules which each may or may not exist in~$G$. Observe that we indicate additional constraints in the last column.

\begin{center}
\setlength{\extrarowheight}{3pt}
\begin{tabular}{ lc | C{5cm} | c }
& Rule & Condition on Subword & Constraint \\
\hline
\rowgap{4}\\
\hline
($R_{1,d}$) & $S \ra \emptyword$ & absence string contains~$d$ & $d \in \Sigma$\\
($R_2$) & $S \ra \emptyword$ & no side condition & ---\\
($R_3$) & $S \ra \emptyword$ & $S^h$ absent & $h$ maximised\\
($R_4$) & $S \ra \emptyword$ & $S^{h'}$ present & $h'$ minimised\\
\end{tabular}
\end{center}

Note that if both~$R_3$ and $R_4$ are present, then one may assume w.l.o.g.\ that~$h' > h$; otherwise,~$S \ra \emptyword$ is applicable regardless of the sentential form containing~$S$).

\medskip

Now, assume~$G$ generates~$L$. We argue by contradiction.

\medskip

The following claim is immediate from the fact that~$d^* \subset L(G)$ for every~$d \in \Sigma$.

\begin{claim}\label{clm1}
Let~$d \in \Sigma$. There exist $m, m'$ with $m+m'>0$, such that
\[
	S \Ra^* d^m S d^{m'}
\]
where~$m, m'$ depend on~$d$.
\end{claim}
We will use this claim to produce contradictions in the main part of this proof. To do so, we will repeatedly focus on the derivation in~$G$ of single characters from~$\Sigma$. For~$d \in \Sigma$, define the derivation~$D_d$ by
\begin{align}\label{eq:deriv}
	D_d \colon S \Ra_1^* S^g \Ra_2 S^i d S^j \Ra_3^* d
\end{align}
where, as in \Cref{clm1}, all of~$g,i,j$ depend on~$d$.
For convenience, we have labelled the derivation steps
\[
	\Ra_1^*, \hspace*{1em} \Ra_2, \text{\hspace*{1em} and \hspace*{1em}} \Ra_3^*
\]
and we refer to them repeatedly below. We also repeatedly use the following fact about the structure of~$\Ra_3^*$ in~$D_d$:

\begin{claim}\label{clm2}
	The derivation step~$\Ra_3^*$ in~$D_d$ is w.l.o.g.\ of one of the following forms:
	\begin{enumerate}
		\item $\Ra_3^*$ is the empty step; \label{s1}
		\item $\Ra_3^*$ is a repeated application of~$R_{1,e}$ for some~$e \in \Sigma - \{ d \}$; \label{s2}
		\item $\Ra_3^*$ is a repeated application of an absence rule~$R_{5,d}$ eliminating instances of~$S$;  \label{s3}
		\item $\Ra_3^*$ is a repeated application of an absence rule~$R_{5,d}$ and a presence rule~$R_{6,d}$, both eliminating instances of~$S$;  \label{s4}
		\item $\Ra_3^*$ is a repeated application of an absence rule~$R_{7,d} \colon S \ra S^{\ge 2}$ whose absence string contains~$e \in \Sigma - \{ d \}$, followed by repeated applications of~$R_{5,d}$ and~$R_{6,d}$. \label{s5}
		\item $\Ra_3^*$ is a sequence of presence rules and absence rules, whose absence strings contain only~$d$ (if any characters from~$\Sigma$), followed by repeated applications of~$R_{5,d}$ and~$R_{6,d}$. \label{s6}
	\end{enumerate}
\end{claim}

\begin{cproof}
	Case~(\ref{s1}) follows in the case that~$i = j = 0$ in~(\ref{eq:deriv}).
	
	For case~(\ref{s2}), note that if~$\Ra_3^*$ uses~$R_{1,e}$ for some~$e \neq d$, then a repeated application of said rule will yield~$d$, by the structure of the left-hand side of~$\Ra_3^*$.
		
	Now, let~$R_{5,d} \colon S \ra \emptyword$ be an absence rule, and suppose case~(\ref{s3}) fails. We show that we must be in case~(\ref{s4}),~(\ref{s5}), or~(\ref{s6}).
	
	First, w.l.o.g.\ we may assume that~$R_{5,d}$ is the last rule applied to eliminate all instances of~$S$. (Note that it may be the case that no rule of the form~$R_{5,d}$ exists.) Since case~(\ref{s3}) fails, there must be a presence rule~$R_{6,d}$ which is applied \emph{before} the last sequence of~$R_{5,d}$-applications has become valid. Note that the presence string of~$R_{6,d}$ and the absence string of~$R_{5,d}$ satisfy the following relationship:
	\begin{quote}
		For any subword in the derivation step~$\Ra_3^*$ containing the presence string of~$R_{6,d}$, there exists an instance of~$S$ after whose elimination:
		\begin{itemize}
			\item $R_{6,d}$ is applicable, or
			\item $R_{5,d}$ is applicable, or
			\item no instances of~$S$ remain.
		\end{itemize}
	\end{quote}
	In the last case, we are clearly done. Similarly,  once~$R_{5,d}$ becomes applicable, we use it repeatedly to eliminate all instances of~$S$, also completing the proof. This gives case~(\ref{s4}).
	
	Finally, it is possible that~$\Ra_3^*$ first \emph{adds} instances of~$S$. To that end, let~$R_{7,d} \colon S \ra S^{\ge 2}$ be an absence rule whose absence string contains~$e \in \Sigma - \{ d \}$. (Note the role of~$d$.) Since~$\Ra_3^*$ eventually eliminates all instances of~$S$, observe that~$R_{7,d}$ must indeed be an absence rule, and that~$R_{7,d}$ must be applicable \emph{before} the repeated applications of~$R_{6,d}$ and, perhaps eventually,~$R_{5,d}$. Indeed, it follows that~$R_{7,d}$ is applicable until the presence string for~$R_{6,d}$ has been generated. This gives case~(\ref{s5}).
	
	If~$R_{7,d}$ is not used, then, by exhaustion, only presence rules or absence rules whose absence strings only contain~$d$ (if any characters from~$\Sigma$) can precede~$R_{6,d}$. This is exactly case~(\ref{s6}).
\end{cproof}

\smallskip

We now carry out an analysis of each case of the derivation~$D_d$ in (\ref{eq:deriv}), each of which will lead to a contradiction through \Cref{clm1}. The cases depend on the productions in our SSCG~$G$.

\subsection*{Case 1: There exist~$R_{1,d}$ and~$R_{1,e}$ for~$d \neq e$, or there exists~$R_2$.}
Since the case for~$R_2$ is analogous, we only show the case for~$R_{1,d}$ and~$R_{1,e}$ below.

Consider $f \in \Sigma - \{d,e\}$, and the derivation
\[
	D_f \colon S \Ra_1^* S^g \Ra_2 S^i f S^j \Ra_3^* f.
\]
Suppose step~$\Ra_2$ uses an absence rule~$R$. If the absence string of~$R$ contains a character in~$\{ d,e \}$, then w.l.o.g.~let it be~$d$. Now, consider the following derivation, whose steps are justified individually: for some~$m + m' > 0$, we have
\begin{align}
	S \underbrace{\Ra^*}_{\text{\Cref{clm1}}} e^m S e^{m'} \underbrace{\Ra^*}_{R} e^mS^{i'} f S^{j'} e^{m'} \underbrace{\Ra^*}_{R_{1,d} \text{ or } R_2} e^m f e^{m'} .\label{eq1}
\end{align}
for some~$i',j'$.

Otherwise, step~$\Ra_2$ uses some presence rule~$R$ with presence string~$S^{g'}$, for some~$g' \leq g$. Assume this for the remainder of Case~1.

Suppose step~$\Ra_1^*$ uses some absence rule $R' \colon S \ra S^{>1}$ with absence string containing a character in~$\{ d,e \}$, say~$d$. Then, for some~$m+m' > 0$ and some~$g''$, we have
\begin{align}
	S \underbrace{\Ra^*}_{\text{\Cref{clm1}}} e^m S e^{m'} \underbrace{\Ra^*}_{R'} e^m S^{g''} e^{m'} \underbrace{\Ra^*}_{R_{1,d}} e^m S^{g} e^{m'} \underbrace{\Ra}_{R} e^mS^i f S^j e^{m'} \underbrace{\Ra^*}_{R_{1,d}} e^m f e^{m'}. \label{eq2}
\end{align}

Otherwise, step~$\Ra_1^*$ does not use any absence rule~$R' \colon S \ra S^{> 1}$ with absence string containing a character in~$\{ d,e \}$. Then, for some~$m + m' > 0$, we have
\begin{align}
	S \underbrace{\Ra^*}_{\text{\Cref{clm1}}} e^m S e^{m'} \underbrace{\Ra^*}_{\Ra_1^* \text{ in } D_f} e^m S^{g} e^{m'} \underbrace{\Ra}_R e^m S^i f S^j e^{m'} \underbrace{\Ra^*}_{R_{1,d}} e^m f e^{m'} . \label{eq3}
\end{align}

\medskip

Since~\Cref{eq1,eq2,eq3} all generate strings of more than one character from~$\Sigma$, Case~1 is proven.

\subsection*{Case 2: There exists~$R_{1,d}$ for exactly one~$d \in \Sigma$.}

Let~$e,f$ be the two other characters in~$\Sigma - \{ d \}$.
We handle a simple case first. W.l.o.g.\ fix~$e$ and suppose in the derivation
\[
	D_e \colon S \Ra_1^* S^g \Ra_2 S^i e S^j \Ra_3^* e
\]
that step~$\Ra_2$ uses some rule~$R$ which is  either an unconditional rule (i.e.~$\alpha = \beta = 0$ in \Cref{dfn:sscg}) or an absence rule whose absence string does not contain~$f$. Then, for~$m + m' \ge 1$ and some~$i',j'$, we have
\begin{align}
	S \underbrace{\Ra^*}_{\text{\Cref{clm1}}} f^m S f^{m'} \underbrace{\Ra}_{R} f^m S^{i'} e S^{j'} f^{m'} \underbrace{\Ra^*}_{R_{1,d}} f^m e f^{m'} \label{eq4}
\end{align}
which generates a string containing two distinct characters from~$\Sigma$, giving a contradiction.

Thus, assume that in the derivation~$D_f$ (respectively~$D_e$), the rule having~$f$ (resp.~$e$) on the right-hand side is rule~$R_{8,f}$ (resp.~$R_{8,e}$) which is either a presence rule whose presence string contains only instances of~$S$, or an absence rule whose absence string contains~$e$ (resp.~$f$). Now, we focus on the steps in the derivation
\[
	D_d \colon S \Ra_1^* S^g \Ra_2 S^i d S^j \Ra_3^* d.
\]
If step~$\Ra_2$ requires an absence rule whose absence string contains a character other than~$d$, then w.l.o.g.\ assume this character is~$e$.

We break down the remaining argument for Case~2 into subcases.

\subsubsection*{Case 2.1: Step~$\Ra_3^*$ requires~$R_{7,d}$.}
Suppose
\[
	S^i d S^j  \, \underbrace{\Ra^*}_{R_{7,d}} \, \alpha \, \underbrace{\Ra^*}_{R_{6,d}} \, d
\]
where $\alpha$~also contains the presence string of~$R_{6,d}$. Assume w.l.o.g.\ that~$\alpha$ starts with~$S$ (if~$\alpha$ ends with~$S$, argue similarly). Since~$R_{7,d}$ only adds instances of~$S$ and so~$\alpha$ does not contain~$e$, the rule~$R_{8,f}$ is applicable in the derivation below, which also uses \Cref{clm2} for~$D_d$:
\begin{align}
	S \underbrace{\Ra^*}_{\Ra_1^* \, + \; \Ra_2} S^i d S^j \underbrace{\Ra^*}_{R_{7,d}} \alpha \underbrace{\Ra^*}_{R_{7,d}} S^p \alpha \underbrace{\Ra^*}_{R_{8,f}} S^{p'} f S^{p''} \alpha \underbrace{\Ra^*}_{R_{6,d}} f \alpha \underbrace{\Ra^*}_{D_d} fd \label{eq5}
\end{align}
where
\[
	p \geq \begin{cases}
		1 + \text{length of the presence string in $R_{8,f}$} & \text{if~$R_{8,f}$ is a presence rule}\\
		1 & \text{otherwise.}
	\end{cases}
\]
Note that the ``otherwise''-case is sufficient as it is assumed that if~$R_{8,f}$ is an absence rule then that absence rule contains the terminal~$e$.

\subsubsection*{Case 2.2: Not Case 2.1 and the step~$\Ra_1^*$ requires a rule~$R$ of the form~$S \ra S^{\geq 2}$, whose absence string contains a character from $\Sigma$} In this case, again using \Cref{clm2} for~$D_d$,
\begin{align}
	S \underbrace{\Ra^*}_{D_d} S^g \underbrace{\Ra^*}_R S^{g+p} \underbrace{\Ra}_{R_{8,f}} S^g S^{p'} f S^{p''} \underbrace{\Ra^*}_{R_{1,d}} S^g f \underbrace{\Ra^*}_{D_d} df \label{eq6}
\end{align}
where, as before,
\[
	p \geq \begin{cases}
		1 + \text{length of the presence string in $R_{8,f}$} & \text{if~$R_{8,f}$ has a presence rule}\\
		1 & \text{otherwise}
	\end{cases}
\]
which suffices, as before, since we assumed that if~$R_{8,f}$ has an absence rule, then that absence rule can only contain the terminal~$e$. 

\subsubsection*{Case 2.3: Not Case~2.1 and Not Case~2.2} In this case, the derivation~$D_d$ only requires absence of~$e$ (that is, any absence rule appearing in this derivation contains either~$e$ or no characters from~$\Sigma$). Thus, \Cref{clm2} for~$D_d$ implies
\begin{align}
	S \underbrace{\Ra^*}_{\text{\Cref{clm1}}} f^m S f^{m'} \underbrace{\Ra^*}_{D_d} f^m d f^{m'} \, . \label{eq7}
\end{align}
(Note that this derivation could also yield~$df$.)

\medskip

Since derivations~(\ref{eq4}) to~(\ref{eq7}) all generate strings of more than one character from~$\Sigma$, Case~2 is also proven.

\subsection*{Case 3: Rules~$R_{1,d}$ and~$R_2$ do not exist for any~$d \in \Sigma$.}
Observe that in this case~$R_3$ must exist in~$G$ since otherwise~$G$ could not generate~$\emptyword$. Further, recall that all rules
\[
	R_{5,d} \colon S \ra \emptyword
\]
are absence rules with absence string~$S^h$.

Further, the presence strings of the rules~$R_{6,d}$ are of the form~$S^{i'} d S^{j'}$, where~$i', j' \le h$, and at most one of~$i', j'$ can be~$h$ (recall that~$S \ra \emptyword$ in presence of~$S^h$ cannot be valid).

\medskip

For~$d \in \Sigma$, we again focus on
\[
	D_d \colon S \Ra_1^* S^g \Ra_2 S^i d S^j \Ra_3^* d \, .
\]

Our aim is to limit the number of distinct absence rules containing a character in absence strings used in~$D_d$. We call such rules \define{special rules}. A rule is \define{$S_d$-special} if it is an absence rule of the form~$S \ra S^{\geq 2}$, and its absence string contains~$d \in \Sigma$. For example,~$R_{7,d}$ is either~$S_e$- or~$S_f$-special, where~$e,f \neq d$.

We do not have much control over the derivation step~$\Ra_2$ in~$D_d$. In step~$\Ra_1^*$ of~$D_d$, note that the absence rules containing a character from~$\Sigma$ in the absence string are~$S_x$-special, as there are no~$R_{1,x}$ rules in the derivation---this follows since neither Case~1 nor Case~2 applies.

\medskip

Below, we modify the derivation of step~$\Ra_1^*$ in~$D_d$ to ensure that at most one~$S_x$-special rule is used. This simplifies our case analysis later.

\subsubsection*{The derivation step $\Ra_1^*$}
If step~$\Ra_1^*$ uses no~$S_x$-special rules, then there is nothing to do.

If Rule~$R_4$ is used, then~$g \geq h'-1$, as there is no valid condition for usage of any rule~$S \ra \emptyword$ which is satisfied by~$S^{h'-1}$, since~$h'-1 \geq h$). Then, using any~$S_x$-special rule~$R$ one obtains
\[
	S \underbrace{\Ra^*}_R S^{g'} \underbrace{\Ra^*}_{R_4} S^g
\]
where~$g' \ge g$.

Note that the applications of~$R_4$ need not be required.
So, suppose $R_4$ is never used.

If some $S_x$-special rule
\[
	R_{10} \colon S \ra S^p
\]
with $2 \leq p < h$ is used, then let~$m \in \{ 1,2,\ldots,p-1\}$ be such that~$g \equiv m \mod (p-1)$.
Then
\[
	S \underbrace{\Ra}_{R_{10}} S^p \underbrace{\Ra^*}_{R_3} S^m \underbrace{\Ra^*}_{R_{10}} S^g
\]
which is as desired.

Now, we are left with only those cases in which
\begin{itemize}
	\item rule $R_4$ is not used, and
	\item only~$S_x$-special rules with~$S \ra S^p$, with~$p \geq h$, are used.
\end{itemize}
Note that~$R_3$ is not used after any of the special rules in this part of the derivation, since~$R_3$ requires the absence of~$S^h$, by definition. If multiple~$S_x$-special rules are used, then define the absence rule
\[
	R_{11} \colon S \ra S^p
\]
with~$p$ minimal, and observe that
\[
	S \underbrace{\Ra^*}_{R_{11}} S^{1+r(p-1)}
\]
for any~$r \geq 0$, and that
\begin{align}
	2h-1 \leq 2p-1 \leq g \leq g+i'+j'+1 \leq 1+r'(p-1) \tag{$*$} \label{ineq}
\end{align}
for large enough~$r'$, where~$i',j'$ are as they appear in~$R_{6,d}$; otherwise, take~$i'=j'=0$.

\medskip

In summary, from the above analysis, we know that either
\[
	S \underbrace{\Ra^*}_{R_{11}} S^g
\]
or
\begin{center}
	$S \underbrace{\Ra^*}_{R_{11}} S^{g'}$ and $S \underbrace{\Ra^*}_{R_{11}} S^{g''}$
\end{center}
and
\[
	2p-1 = g' \le g \le g'' = r' (p-1) + 1
\]
which follows from~(\ref{ineq}). Furthermore, these derivations use at most one $S_x$-special rule each.

\medskip

In the case~$g'<g$, we have~$g \geq 2h$, and
thus~$i \ge h$ or~$j \ge h$. Therefore, the derivation~$\Ra_3^*$ uses rule~$R_{6,d}$. We may also choose the instance of~$S$ in derivation step~$\Ra_2$ so that both~$i \ge h-1$ and~$j \ge h-1$. Thus,~$R_{6,d}$ can always be made applicable in~$S^i d S^j$, and even if we choose~$g'$ or $g''$ instead of $g$, if the rule~$R$ used in~$\Ra_2$ is applicable to~$S^{g'}$ ($S^{g''}$, respectively), then the same derivation~$S^{g'} \Ra^* d$ ($S^{g''} \Ra^* d$, respectively) can be carried out, too.
Note:
\begin{itemize}
	\item if~$R$ is a presence rule, then it is also applicable to~$S^{g''}$, and
	\item if~$R$ is an absence rule, then it is also applicable to~$S^{g'}$.
\end{itemize}
Thus, we may assume w.l.o.g.\ that in~$D_d$, $g$ is modified to~$g'$ or~$g''$, based on whether we have~$S^{g'} \Ra^* S^i d S^j$ or~$S^{g''} \Ra^* S^i d S^j$.

In conclusion, we know that~$\Ra_1^*$ and~$\Ra_3^*$ use at most one absence rule each, which is~$S_x$-special for some~$x$. Furthermore, the~$S_x$-special rule used in~$\Ra_1^*$ can be assumed to be the same for all $D_x$ (with~$g$ depending on~$x \in \Sigma$).

\medskip

We now complete the argument of this Case~3 by another case analysis.
First, assume:
\begin{itemize}
	\item[$\diamond$] no~$S_x$-special rules are used
\end{itemize}
Then, we may choose any~$d$ and consider again
\[
	D_d \colon S \Ra_1^* S^g \Ra_2 S^i d S^j \Ra_3^* d \, . \label{dag} \tag{$\dagger$}
\]
If~$\Ra_2$ is an absence rule and contains some character in~$\Sigma$, assume w.l.o.g.\ that it is~$e$. Let~$f \in \Sigma - \{ d,e \}$. Then, for~$m + m' > 0$, \Cref{clm2} implies
\begin{align}
	S \underbrace{\Ra^*}_{\text{\Cref{clm1}}} f^m S f^{m'} \underbrace{\Ra^*}_{D_d} f^m d f^{m'} . \label{eq8}
\end{align}

\medskip

Next, we assume:
\begin{itemize}
	\item[$\diamond$] there exists an~$S_d$-special rule~$R_{12}$ in~$G$ which is the only~$S_x$-special rule appearing in~$\Ra^*_1$ of~(\ref{dag}).
\end{itemize}

In~(\ref{dag}), suppose~$\Ra_2$ employs rule~$R$, which is either a presence rule or an absence rule with absence string containing a character in~$\Sigma-\{d\}$; w.l.o.g.\ assume this is~$e$. We may assume that if any $S_x$-special rule is used in~$\Ra_3^*$, then its use is made void by using~$R_{12}$ in~$S^g \Ra^* S^{\geq g+2h}$; this is done by splitting~$S^{\ge 2h}$ into sufficiently large blocks either side of~$d$ so that~$R_{6,d}$ becomes applicable.
In particular, for~$m + m' > 0$,

\begin{align}
	S \underbrace{\Ra^*}_{\text{\Cref{clm1}}} f^m S f^{m'} \underbrace{\Ra^*}_{R_{12}} f^m S^{\geq g+2h}f^{m'} \underbrace{\Ra}_R f^m S^{i+h} d S^{\geq j+h}f^{m'} \underbrace{\Ra^*}_{R_{6,d} \text{ or } R_{5,d}} f^m d f^{m'} \label{eq9}
\end{align}
where the last step uses either~$R_{6,d}$ until all~$S$ are eliminated, or~$R_{5,d}$, if applicable (cf.~\Cref{clm2}).

\medskip

Finally, the only case left is that in which:
\begin{itemize}
	\item[$\diamond$] the step $\Ra_2$ uses an absence rule~$R$ containing no characters from~$\Sigma - \{ d \}$.
\end{itemize}
In that case, if~$\Ra_3^*$ uses an~$S_x$-special rule, and its absence string has a character in~$\Sigma - \{ d \}$, assume w.l.o.g.\ it is~$e$. Then, for~$m + m' > 0$,
\begin{align}
	S \underbrace{\Ra^*}_{\text{\Cref{clm1}}} f^m S f^{m'} \underbrace{\Ra^*}_{D_d} f^m S^{g} f^{m'} \underbrace{\Ra}_{R} f^m S^{i} d S^{j}f^{m'} \underbrace{\Ra^*}_{D_d} f^m d f^{m'} . \label{eq10}
\end{align}

\medskip

Since derivations~(\ref{eq8}) to~(\ref{eq10}) all generate words with distinct characters from~$\Sigma$, the result follows.

\bigskip

The above exhaustive case analysis shows that no SSCG can recognise our fixed ternary language~$L$. This completes the proof.
\end{proof}

Since all unary regular languages need only one non-terminal (\Cref{thm-unaryregular}) while there exists a ternary regular language which requires two non-terminals (\Cref{thm-regular-needtwo}), the answer to the following question would prove optimality:

\begin{question}
Can all regular (or linear) languages over a binary alphabet be recognised by a simple semi-conditional grammar with only one non-terminal?
\end{question}

\section{Complexity of SSCGs With One Non-Terminal}
In this section, we consider the complexity strength of SSCGs with one non-terminal. We first show that such SSCGs can generate non-context free languages.

\begin{thm}\label{thm-oneNT-nonCF}
Let~$\Sigma = \{ 0,1,2,3,4,5 \}$. There is a non-context-free language on~$\Sigma$ which can be generated by a simple semi-conditional grammar with only one non-terminal.
\end{thm}

\begin{proof}
The idea is to build a language $L$ such that
\begin{align}
	L \cap \{0\}^+ \{1\} \{23\}^+ \{4\} \{5\}^+ = \set{0^n 1 (23)^n 4 5^n}{n\text{ is odd}}. \tag{$*$} \label{cond}
\end{align}
Since context-free languages are closed under intersections with regular languages~\cite[p.~135]{HopUll79}, it follows from the pumping lemma for context-free languages that~$L$ is not context-free.

We define our SSCG~$G$ as follows:
\begin{center}
\setlength{\extrarowheight}{3pt}
\begin{tabular}{rc|c|c}
& Rule & Condition on Subword & Explanation \\
\hline
\rowgap{4}\\
\hline
(I) & $S \ra 0S23S5$ & $0$ absent & initialisation \\
\hline
(L1) & $S \ra 0S3$ & $3S5$ present & loop \\
(L2) & $S \ra 2S5$ & $0S3$ present & loop \\
(L3) & $S \ra 0S2$ & $2S5$ present & loop \\
(L4) & $S \ra 3S5$ & $0S2$ present & loop \\
\hline
(T1) & $S \ra 1$ & $3S5$ present & termination \\
(T2) & $S \ra 4$ & $012$ present & termination \\
\end{tabular}
\end{center}

\subsection*{Verification}
We verify that~$L(G)$ satisfies~(\ref{cond}).
It is easily seen that any generation of a word must begin with the initialisation step, as none of the other rules are applicable otherwise. It is now readily verified that the sequence
\[
	(\text{I}) - (\text{L}1) - (\text{L}2) - (\text{L}3) - (\text{L}4)
\]
generates the word
\[
	0^3 S (23)^3 S 5^3
\]
where each loop step is applied to that $S$-instance which does \emph{not} appear in the relevant side condition. By induction, the sequence
\[
	(\text{I}) - \big( (\text{L}1) - (\text{L}2) - (\text{L}3) - (\text{L}4) \big)^n - (\text{T}1) - (\text{T}2)
\]
generates the word
\[
	0^{2n+1} \,1 \, (23)^{2n+1} \, 4 \, 5^{2n+1}
\]
which proves the $\supseteq$-inclusion in~(\ref{cond}).

To prove equality, first note that the loop can only be carried out in order. Further, by a case analysis for all four loop stages, it is seen that an early termination, by applying (T1) and (T2) before completing a full loop, does not generate a word in~$\{0\}^+ \{1\} \{23\}^+ \{4\} \{5\}^+$. The same occurs in case that any (L$i$) or (T$i$) is applied to an $S$-instance different from the intended one. This completes the argument.
\end{proof}

By coding, the language~$L$ above can be constructed over a binary alphabet, but we note that that would require a larger window width for the required or forbidden words.

\subsection{Decidability of One-Non-Terminal SSCGs}
In contrast to the previous theorem, there is a clear limitation to the complexity of languages which can be generated by SSCGs with only one non-terminal.

\begin{thm}\label{thm-oneNT-linearspace}
If there is only one non-terminal~$S$ in a simple semi-conditional grammar then the language it generates is decidable, even in non-deterministic linear space.
\end{thm}

To give the proof, we require an additional lemma. First, we fix some notation. Let $G$ be an SSCG of one non-terminal, which we denote by~$S$. Let $m$ be the length of the largest side condition. Let $p$ be the length of the largest output of any of the rules in $G$. Assume $p>1$, as grammars with $p = 1$ are trivial. Fix an alphabet $\Sigma$.

\begin{lem}\label[lem]{lem:appl}
	Suppose $u,v \in (\Sigma \cup \{ S \})^*$ and $n > m$. Then a production rule $\ell$ is applicable to~$wS^m u$ if and only if it is applicable to~$wS^n u$.
\end{lem}

\begin{proof}
	The subwords of~$wS^mu$ and~$wS^nu$ of length at most~$m$ are the same.  
\end{proof}

The following lemma forms the core of the proof of \Cref{thm-oneNT-linearspace}.

\begin{lem}\label[lem]{lem:good}
	Suppose~$w \in \Sigma^n$ is generable from~$G$. Then there exists~$k \in \mathbb{N}$ such that~$w$ has a $G$-derivation~$\tau$ of length~$k$
\[
	\tau \colon \{ 0,1,\dots,k-1 \} \to (\Sigma \cup \{ S \})^*
\]
where~$\tau(i) \ra \tau(i+1)$ is one step of the derivation, and which satisfies
\[
\tau(0) = S \qquad \text{ and } \qquad \tau(k-1) = w
\]
such that, for every~$i < |\tau| = k$, we have
	\[
		|\set{a \in \tau(i)}{a = S}| \le (m+p-1)(n+1) + pn \, .
	\] 
\end{lem}

Observe that this bound is linear in~$n$, and that~$m$ and~$p$ are fixed in~$G$.

\begin{proof}
	Fix some~$\sigma$ which generates~$w$ from~$G$. We define a new derivation sequence~$\tau$ which is as desired. By definition of grammars, we have~$\sigma(0) = \tau(0) = S$, and by the end we also have~$\sigma(|\sigma|-1) = \tau(|\tau| -1) = w$. In fact, we will impose that~$|\sigma| = |\tau|$, for ease of presentation. The reader may consider this to be an application of a new rule,~$S \ra S$. Of course, such an extended derivation sequence can later be shortened by removing the repeated steps.
	
	\medskip
	
	Fix~$i < |\sigma|$.
	
	\medskip
	
	For any occurrence of a symbol~$a \in \Sigma$~in $\sigma(i)$, by considering the remainder of the derivation sequence~$\sigma$ after step~$i$, we may associate this occurrence of~$a$ with a corresponding occurrence of~$a$ in~$w$. This gives us some $j < n$ with~$w(j) = a$, which is the occurrence’s \define{final position} in~$w$.
	Suppose that
	\begin{align}
		\sigma(i) = S^{r_0} a^i_0 S^{r_1} a^i_1 S^{r_2} \cdots S^{r_{k-1}} a^i_{k-1} S^{r_k} \label{seww}
	\end{align}
	whose sequence~$(r_s)$ depends on~$i$.
	Let~$j^i_s$ denote the final position of~$a^i_s$ in~$w$.
	\begin{figure}[htp]
		\begin{center}
		\setlength{\extrarowheight}{3pt}
		\begin{tabular}{c | p{0.35cm}ccccc}
			Stage && \multicolumn{5}{c}{$\sigma(i)$}\\
			\hline
			\rowgap{5}\\
			\hline
			$i-1$ & && $\cdots$ & $S$ & $\cdots$ \\
			$i$ & & & $\cdots$ & $a^i_s$ & $\cdots$ & \\
			$\vdots$ &&&& $\Big\downarrow$ &\\
			$|\sigma| - 1$ && $\cdots$ & $w(j^i_s - 1)$ & $a^i_s$ & $w(j^i_s + 1)$ & $\cdots$
		\end{tabular}
		\end{center}
		\caption{The final occurrence of some~$a^i_s$ in~$w$, which appeared in derivation step~$i$.}
	\end{figure}%

	Throughout the construction of~$\tau$, we will preserve the structure of the words~$\sigma(i)$ in~$\tau$: for~$\sigma$ as in~(\ref{seww}), we will define~$\tau$ so that
	\[
		\tau(i) = S^{r'_0} a_0^{i} S^{r'_1} a_1^{i} S^{r'_2} \cdots S^{r'_{k-1}} a_{k-1}^{i} S^{r'_k}
	\]
	whose sequence~$(r'_s)$ also depends on~$i$, and so that~$k \le n$.
	
	The remainder of the proof hinges on the structure of this derivation~$\tau$, which we study via the indices~$j^i_s$. For unity of presentation, it is helpful to define
	\begin{align*}
		j^i_{-1}  = -1 \quad \text{ and } \quad j^i_k = n\, .
	\end{align*}
	Our goal is to construct a derivation sequence~$\tau$ whose coefficients~$r'_s$ satisfy the following relation, defined inductively on~$i$ (to keep the notation moderately simple, we now drop the superscript~$i$). For every~$s < k$ (as per~(\ref{seww})):
	\begin{enumerate}
		\item[(I1)] If~$r_s < m(j_s - j_{s-1})$ then~$r_s' = r_s$.
		\item[(I2)] If~$r_s \ge m(j_s - j_{s-1})$ then~$r_s \ge r'_s \ge m(j_s - j_{s-1})$.
	\end{enumerate}
	In particular,
	\[
		r_s \ge r'_s \hspace*{1em} \text{for all~$s < k$}
	\]
	and
	\begin{align}
		\text{if } r_s > r_s' \quad \text{ then } \quad r_s > r_s' \ge m(j_s - j_{s-1}) \, . \label{eq:asttt}
	\end{align}
	We now construct the derivation sequence~$\tau$---and hence the sequence~$(r_s')$---from~$\sigma$ while maintaining the inductive hypothesis.
	
	\bigskip
	
	We begin with~$\tau(0) = \sigma(0) = S$, and clearly our inductive assumption is maintained.
	
	\bigskip

	Suppose~$\sigma(i+1)$ is given and~$\tau(i)$ already exists. Observe that, by induction and \Cref{lem:appl}, the production rules applicable to~$\sigma(i)$ are also applicable to~$\tau(i)$.
The following is readily verified.
	
	\begin{claim}
		The productions applicable to~$\sigma(i)$ are of one of the following three forms:
		\begin{itemize}
			\item $S \ra \emptyword$
			\item $S \ra S^t$ for some~$t > 0$
			\item $S \ra S^{t} u S^{t'}$ where~$t,t' \ge 0$ and~$u$ begins and ends with a non-terminal
		\end{itemize}
	\end{claim}
	
	Using the claim, we now structure our inductive argument to construct~$\tau$ obeying the inductive hypotheses~(I1) and~(I2). Our action depends on the production~$\ell$ applied to~$\sigma(i)$:
	
	\begin{enumerate}
		\item \underline{$\ell = S \ra \emptyword$.} Suppose~$\ell$ is applied to the block~$a_{s-1} S^{r_{s}} a_s$ in~$\sigma(i)$. Thus,~$\sigma(i + 1)$ contains
		\[
			a_{s-1} S^{r_{s} - 1} a_s.
		\]
		\begin{itemize}
			\item If~$r_s = r'_s$, then apply~$\ell$ to~$\tau(i)$ in exactly the same way.
			\item If~$r_s > r'_s$, then let~$\tau(i+1) = \tau(i)$.
		\end{itemize}
		\item \underline{$\ell = S \ra S^t$.} Suppose~$\ell$ is applied to the block~$a_{s-1} S^{r_{s}} a_s$ in~$\sigma(i)$. Thus,~$\sigma(i + 1)$ contains
		\[
			a_{s-1} S^{r_{s} +t-1} a_s.
		\]
		\begin{itemize}
			\item If~$r_s = r'_s$ and~$r_s < m(j_s - j_{s-1})$, then apply~$\ell$ to~$\tau(i)$ in exactly the same way.
			\item Otherwise, let~$\tau(i+1) = \tau(i)$.
		\end{itemize}
		\item \underline{$\ell = S \ra S^t u S^{t'}$.} We assume that $u$ is a word beginning and ending with terminal symbols. We will generate~$\tau(i+1)$ by applying the same rule to~$\tau(i)$. Some care must be taken in choosing which~$S$ form~$\tau(i)$ to apply the rule to, since we must maintain the relationship between every~$r$ and~$r'$ at step~$i+1$.
		
		Suppose that~$\ell$ is applied to an instance of~$S$ in the block~$a_{s-1} S^{r_{s}} a_s$ in~$\sigma(i)$, which we may hence write as~$\alpha \, a_{s-1} S^{r_s} a_s \, \beta$ for some~$\alpha, \beta$. Thus,~$\sigma(i + 1)$ is of the form
		\[
			\alpha \, a_{s-1} S^q u S^{\tilde{q}} a_s \, \beta \, .
		\]
		Observe that~$q + \tilde{q} = r_s + t + t' - 1$.
		\begin{itemize}
			\item If~$r_s = r'_s$, then apply~$\ell$ to~$\tau(i)$ in exactly the same way as in~$\sigma(i)$ (i.e.~choose the same instance of~$S$ in~$\tau(i)$).
			\item If~$r_s > r'_s$ and~$q$ is sufficiently small so that case (I1) applies to its block inside~$\sigma(i+1)$, we choose an~$S$ in~$\tau(i)$ such that~$\tau(i+1)$ contains
			\[
				a_{s-1} S^{q} u S^{\tilde{q}'} a_s \, .
			\]
	We will necessarily have~$\tilde{q} \ge \tilde{q}'$ since~$r_s > r'_s$, but note that by additivity and the inductive assumption on~$r'_s$,~$\tilde{q}$ satisfies (I2) and $\tilde{q}'$ will be sufficiently large.
		\item If~$r'_s < r_s$ and $\tilde{q}$ is sufficiently small that it will be in case~(I1) for~$\sigma(i+1)$, we do as in the previous case, mutatis mutandis.
		\item If~$r'_s < r_s$ and~$q$ and~$\tilde{q}$ are both sufficiently large that they will be in case~(I2) for~$\sigma(i+1)$, we choose an~$S$ in~$\tau(i)$ such that~$\tau(i+1)$ contains the subword
		\[
			a_{s-1} S^{q'} u S^{\tilde{q}'} a_s
		\]
		where~$q' \le q$ and~$\tilde{q}' \le \tilde{q}$, and both are sufficiently large. This is possible by additivity and the inductive assumption on~$r'_s$.
		\end{itemize}
	\end{enumerate}
	
	\begin{claim}
		The inductive hypothesis on the various~$r'_s$ is maintained, for all~$i < |\sigma|$. 
	\end{claim}
	
	\begin{cproof}
		This is immediate from the construction.
	\end{cproof}
	
	It follows that~$\tau(|\tau| - 1) = \sigma(|\sigma| - 1) = w$, and each~$\tau(i+1)$ is generated from~$\tau(i)$ by the application of a legal rule from~$G$ or by simply doing nothing, so~$\tau$ is a (generalised) derivation sequence for~$w$.
	
	With~$\tau$ constructed as above, the following claim gives the required bound.
	
	\begin{claim}
		For $\tau(i) = S^{r'_0} a_0 S^{r'_1} a_1 \cdots S^{r'_{k-1}} a_{k-1} S^{r'_k}$, we have
		\begin{align}
			\sum_{s = 0}^k \max \big( 0, r_s' - m(j_s - j_{s-1}) - p + 1 \big) \le kp \, . \label{eq:lkj}
		\end{align}
	\end{claim}
	
	\begin{cproof}
		Since~$\tau$ is constructed by induction on~$i$, we prove this claim also by induction. Note that~(\ref{eq:lkj}) holds for~$i=0$ since we postulated that~$\sigma(0) = \tau(0) = S$.
		
		If~$\tau(i)$ and~$\tau(i+1)$ have the same terminal symbols, then either~$\tau(i+1) = \tau(i)$, or~$\tau(i+1)$ was formed from~$\tau(i)$ by applying a rule of the form~$S \ra \emptyword$, or~$\tau(i+1)$ was formed from~$\tau(i)$ by applying a rule of the form~$S \ra S^t$ with~$t \le p$.
		
		In the first two cases, the inequality is clearly preserved from~$\tau(i)$ to~$\tau(i+1)$.
		
In the last case, note that by construction, the~$r'_s$ from~$\tau(i)$ to which we are applying the rule must have
\[
	r'_s < m(j_s - j_{s-1}) \, .
\]
Then, the corresponding~$r'_s$ in~$\tau(i+1)$ will have
\[
	r'_s < m(j_s - j_{s-1}) + p
\]
and so it will not participate in the sum for~$\tau(i+1)$. Thus, the non-zero summands for~$\tau(i+1)$ are the same as those for~$\tau(i)$, meaning the inequality is preserved.

If~$\tau(i+1)$ was formed from~$\tau(i)$ by the application of a rule of the form~$S \ra S^{t} u S^{t'}$, then consideration of the subcases will reveal that the total amount by which the sum has increased is bounded above by~$ t + t' < p$. Since~$\tau(i+1)$ contains at least one more terminal symbol than~$\tau(i)$, the inequality is preserved.
	\end{cproof}
	
	Since~$\sum_{s = 0}^k (j_s - j_{s-1}) = - j_{-1} + j_k = n+1$, (\ref{eq:lkj}) implies
	\[
		\sum_{s=0}^k r'_s \le m(n+1) + (k+1)(p-1) + kp \, . \qedhere
	\]
\end{proof}

\Cref{thm-oneNT-linearspace} now follows immediately: by \Cref{lem:good}, to determine membership in the language generated by~$G$, a search over derivation sequences bounded linearly in the number of non-terminals is sufficient. This completes the proof of \Cref{thm-oneNT-linearspace}.

\begin{question}
	Can \Cref{thm-oneNT-linearspace} be improved, either by weakening hypotheses or by strengthening the complexity conclusion?
\end{question}

\section{Membership Testing of SSCGs Is $\NP$-Hard}
The next result shows that testing membership in SSCG languages is $\NP$-hard even for SSCGs with only two non-terminals.

\begin{thm}\label{thm-nphardness}
The following problem is $\NP$-hard.
\begin{quote}
	Given a simple semi-conditional grammar~$G$ of two non-terminals, $2n+4$ terminals,~$O(n)$ rules, and an~$O(n^3)$ length word~$u$, can~$G$ generate~$u$?
	\end{quote}
If the given grammar~$G$ produces with every rule application at least one terminal, then the problem is $\NP$-complete. 
\end{thm}

\begin{rem}
This hardness is present even in the absence of $\emptyword$-rules in the grammar. If, furthermore, the grammar contains at least one terminal on the right-hand side of every rule, then the problem is $\NP$-complete. The below proof satisfies this hardness relation; the $\NP$-completeness part of the claim follows from the fact that one can guess the derivation. Further, the number of steps is bounded by the length of the input word. All conditions, as well as the applicability of rules, can be checked in polynomial time.
\end{rem}

\begin{proof}[Proof of~\Cref{thm-nphardness}]
We reduce SAT for CNF formulas to the problem stated in the theorem, thus proving $\NP$-hardness.

\medskip

Fix~$n$.

\medskip

Below, we define an SSCG~$G$ with the following two properties:
\begin{enumerate}
	\item For every SAT formula~$P$ there exists a word~$u_P \in L(G)$ which uniquely codes~$P$ in~$L(G)$.
	\item Every SAT formula~$P$ which is coded in~$L(G)$ is satisfiable. 
\end{enumerate}
The SSCG~$G$ will also generate words which do not code SAT formulas as per our coding above. Hence, we denote the set of \define{coded formulas} in~$L(G)$ by~$\tilde{L}(G)$, and we note that
\[
	\tilde{L}(G) \subset L(G) \, .
\]
This is not an issue to prove $\NP$-hardness, though, since the sublanguage of all syntactically correctly coded SAT formulas in~$L(G)$ is regular, as the following claim shows:

\begin{claim}\label{clmm1}
	$\tilde{L}(G) = L(G) \cap R$ for some regular language~$R$.
	In other words, there exists a regular expression which separates the syntactically correct from the syntactically incorrect instances.
\end{claim}

\begin{cproof}
	With~$n$ fixed, a suitable regular expression is readily established via the coding given below.
\end{cproof}
To explain the coding, we first define the language. Denote the terminals by
\[
	x_1,\ldots,x_{n+1}, y_1,\ldots,y_{n+1},\sep,\spp,\fin,\stpp
\]
Each SAT instance is made up of clauses, which are each followed by a separator pair~$\gapp$ (the symbols~$\sep$ avoid~$\emptyword$ appearing on the right-hand side of productions). The instance itself is concluded by the pair~$\stopp$. The end of all variable listings is highlighted by the symbol~$\fin$. Further, we use the following conventions to express SAT formulas:
\begin{itemize}
	\item All literals are represented by a single terminal, and the clauses $x_i \vee y_i$ are omitted.
	\item The disjunctions between literals are omitted.
	\item For all~$k$, we define~$y_k = \lnot \, x_k$.
	\item Every instance is concluded by a listing of all (positive) variables~$x_1,\dots,x_{n+1}$.
\end{itemize}

As a result of our coding (see \Cref{tb:TT} below),~$x_{n+1}$ cannot be incorporated into any SAT formula, but is instead required to terminate. Hence, we use~$n + 1$ many variables to code~$n$-variable SAT formulas.
For example, a typical instance of a coded 3-variable SAT formula (i.e. with~$n=3$) is given by
\[
   x_1x_3\gapp x_2y_3\gapp x_1x_2x_3x_4 \, \fin \stopp
\]
(note the presence of~$x_4$) which codes the SAT formula
\[
	(x_1 \lor x_3) \land (x_2 \lor \lnot x_3) \, . 
\]
If the separator policy is violated, the formula is incorrect.

\medskip

In \Cref{tb:TT} below, we construct the required SSCG. Since the number of variables in our SAT formulation is fixed as~$n$, the SSCG's loop rules are valid for~$k=1,2,\ldots,n+1$.
\begin{center}
\setlength{\extrarowheight}{3pt}
\begin{table}
\begin{tabular}{rc|c|c}
& Rule & Condition on Subword & Explanation \\
\hline
\rowgap{3}\\
\hline
(I1) & $S \ra TT \, \spp \, S$ & $x_1$ absent & initialisation \\
(I2) & $S \ra x_1S \stpp$  & $x_1$ absent & initialisation \\
(I3) & $S \ra x_1SS \stpp$  & $x_1$ absent & initialisation \\
\hline
(L1) & $S \ra x_kS \upgap x_kSS$ & $x_{k-1}S \stpp$ present & loop \\
(L2) & $S \ra x_k \upgap x_kS$ & $x_{k-1}SS \stpp$ present & loop \\
(L3) & $T \ra y_k T \upgap y_k \upgap x_k T$ & $x_kS \stpp$ present & loop \\
(L4) & $T \ra x_k T \upgap x_k \upgap y_k T$ & $x_kSS \stpp$ present & loop \\
\hline
(T1) & $S \ra \fin$ & $x_{n+1}SS \stpp$ present & termination \\
(T2) & $S \ra \fin S$ & $x_{n+1}S\stpp$ present & termination \\
(T3) & $S \ra \sep$ & $\fin S$ present & termination \\
(T4) & $T \ra \sep$ & $S$ absent & termination \\
\end{tabular}
\vspace*{1em}
\caption{We indicate the options in rules~(L$i$) by~(L$i$-$j$); for instance, the rule~(L1-2) is~``$S \ra x_kSS$''.}
\label{tb:TT}
\end{table}
\end{center}

\subsection*{Explanation}
Suppose~$P$ is a SAT formula. We show how to generate~$u_P$ in~$G$.

First, rule~(I1) is applied as many times as needed to generate the required number of clauses in~$P$. We call these the \define{coded clauses}. For instance, we can generate four coded clauses by applying~(I1) four times, yielding
\[
	S \Ra^* TT \, \spp \, TT \, \spp \, TT \, \spp \, TT \, \spp \, S \, .
\]

During generation, a double~$TT$ indicates a clause not yet satisfied, while a single~$T$ indicates a clause which is satisfied. New variables are instantiated once, either \emph{positively} as in~(L1-2) and~(L2-2); or \emph{negatively}, as in~(L1-1) and~(L2-1). Whenever a positively instantiated variable is included in a clause as a positive, then we use~(L4-2) to do so. This eliminates one~$T$-instance in the clause, hence coding that this clause is satisfied. We do the same for negatively instantiated variables which appear as a negative, this time using~(L3-2). Observe that these are the only ways to eliminate the double~$TT$ in a syntactically correct word.

We can now deduce a satisfying valuation for~$P$ from the derivation sequence for~$u_P$, by specifying for~$k = 1,2,\dots,n$:
\begin{equation}\label{eq:split}
	\begin{split}
		x_k = 1 &\iff x_k SS \stpp \text{ appears in the derivation}\\
		y_k = 1 &\iff x_k S \stpp \text{ appears in the derivation}
	\end{split}
\end{equation}
It is verified by construction that no word which codes a SAT formula can contain both~$S \stpp$ and ~$SS\stpp$. Hence, this defines a unique valuation map---but this does not yet show that this valuation map is a satisfying assignment. This is what we establish now.

\begin{claim}
	$\tilde{L}(G) = \set{u_P}{P \text{ is satisfiable}}$, or in other words:
	\begin{enumerate}
		\item All syntactically correct instances are satisfiable with respect to at least one variable assignment.
		\item Unsatisfiable instances cannot be generated, unless they are syntactically incorrect.
	\end{enumerate}
\end{claim}

\begin{cproof}
	We prove~(2) first. Suppose~$P$ is coded in~$\tilde{L}(G)$ via the word~$u_P$. We show that~$P$ is satisfiable. In particular,~$P$ will be satisfied by the satisfaction definition as given in~(\ref{eq:split}) above, denoted by~$v$. To see this, suppose there exists a clause in~$P$ not satisfied by~$v$, which we assume w.l.o.g.\ to be of the form
	\begin{align}
		x_1 \lor x_2 \lor x_3 \, . \label{eq:sat}
	\end{align}
	Since~$v$ does not satisfy the clause, the derivation of~$u_P$ must contain the words
	\[
		x_1 S \stpp \, , \hspace*{2em} x_2 S \stpp \, , \hspace*{2em} \text{and} \hspace*{2em}  x_3 S \stpp \, .
	\]
	By the rules of the grammar, the only way in which~$v$ could not possibly satisfy~(\ref{eq:sat}) is by applying the third option in rule~(L3) for all three variables~$x_1$, $x_2$, and $x_3$. But then, the following derivation of subwords inside~$u_P$ must be present:
	\[
		TT \spp \Ra x_1TT \spp \Ra x_1x_2TT \spp \Ra x_1x_2x_3TT \spp
	\]
	This derivation cannot be concluded to yield a syntactically correctly coded SAT formula, which contradicts the fact that~$u_P \in \tilde{L}(G)$. The same argument applies for longer clauses, or those including both positive and negative variables.
	
	\medskip

	To prove~(1), we show that every satisfiable SAT formula in~$n$ variables can be generated. First, observe that for any SAT formula~$P$:
	\begin{equation}\label{eq:coded}
		\begin{gathered}
	 	\text{$P$ is coded in~$\tilde{L}(G)$}\\
		\iff\\
		\text{In each coded clause~$C$ in~$P$, one instance of~$T$ is eliminated}\\
		\iff\\
		\text{$P$ has a satisfying valuation}
		\end{gathered}
	 \end{equation}
	 Now, let~$P$ be a SAT formula of~$k$ clauses
	\begin{align*}
		P &= \bigwedge_{i < k} C_i
	\intertext{where each clause~$C_i$ has~$m_i$ variables}
		C_i &= \bigvee_{j < m_i} x_{i,j} \, .
	\end{align*}
	Suppose~$v$ is a satisfying valuation. By definition, each~$C_i$ contains a positive (or negative) variable~$x$ which is made true by~$v$. When~$x$ is coded into the coded clauses inside the word~$u_P$, rule~(L4-2) (or rule~(L3-2) in the negative case) is applied, hence eliminating a~$T$-instance from the coded clause~$C_i$. Since this is true for every clause, the formula~$P$ can be generated by~(\ref{eq:coded}), as required.
\end{cproof}

$\NP$-hardness is now immediate, since the transformation
\[
	P \mapsto u_P
\]
is computable in linear time. Together with \Cref{clmm1}, this suffices.
Since the grammar~$G$ produces with every rule application at least one terminal, membership of words in~$L(G)$ is in~$\NP$; one can guess the derivation. Consequently, if the given grammar~$G$ in the statement of the theorem is assumed to produce with every rule application at least one terminal, then the problem is $\NP$-complete.
\end{proof}

Note that the bound $O(n^3)$ in the theorem follows from an \textit{a priori} estimate of the maximum number of clauses in a SAT formula on~$n$ variables without clause repetition.

\section{RE Languages}

We complete this paper by showing the strength of SSCGs with three non-terminals.

\begin{thm}\label{thm-relang}
Every RE language is generated by some SSCG with three non-terminals.  
\end{thm}

\begin{proof}
Let~$L$ be RE. We describe the construction of our SSCG~$G$ which has three non-terminals $S,T,U$.

\medskip

Consider a non-deterministic counter automaton~$A$ with~$r$ counters. We will interpret it as a language-generating device as follows. Based on the current state, the counter automaton can do one of the following:
\begin{enumerate}
	\item check the status of a particular counter (zero/non-zero), or 
	\item increment/decrement a counter, or 
	\item print a terminal symbol on the one-way write-only tape.
\end{enumerate}
The counter automaton~$A$ then transitions to the next state which in case~(1) above depends on the result of the check.

\medskip

In any derivation of~$u \in L$ in our SSCG~$G$, the sequence of sentential forms will mimic the simulation of the counter automaton. (Note that one-step simulations in~$A$ may need several steps in~$G$.)

\medskip

For ease of presentation, we will first describe the derivation for strings~$u \in L$ of length at least
\[
	1 + (r+1)(r+2)/2
\]
where the counter automaton has already printed the first~$(r+1)(r+2)/2$ characters of~$u$. We consider the more general case later.

We will use certain sequences of~$T,U$ in sentential forms as codes for 
\begin{itemize}
	\item states and dummy states (which are used for a sequence of steps in the derivation in~$G$ to implement one-step simulations of~$A$) and
	\item messages which code the presence or absence of certain patterns in the current sentential form.
\end{itemize}

Each code is of the form
\begin{align}
	\left( U^4T \Big( \big\{ U^4,U^5 \big\}T \Big)^3U^4TU^6T \right)^{r'}U^7 T \label{cc:code}
\end{align}
where~$r'$ is sufficiently large to code all messages, states, and dummy states.

We also require that
\begin{align}
	\left( U^4T \Big( \big\{ U^4,U^5 \big\}T \Big)^3U^4TU^6T \right)
\end{align}
has exactly one~$U^5$ in each of its~$r'$ many occurrences. Hence, all codes have the same length.

The following claim is proven by induction on the complexity of words.

\begin{claim}
	Suppose~$c$ is a code as in (\ref{cc:code}). Let~$c'$ be the result of applying exactly one of the following changes to~$c$:
	\begin{itemize}
		\item A finite number of~$T$-instances are changed into $U$-instances.
		\item A finite number of~$U$-instances are changed into~$\emptyword$-instances.
		\item A finite number of~$U$-instances are changed into~$TT$-instances.
	\end{itemize}
	Then~$c'$ is not a code. 
\end{claim}

We now work in~$G$.

\bigskip

At the start of any simulation of a step of the counter automaton~$A$ (except at initialisation time), the sentential form in~$G$ will be of the following form:
\begin{align}\label{eq:codeform}
	T\Sigma^1 T^{\ell+C_1} \Sigma^2 T^{\ell+C_2} \cdots \Sigma^{r} T^{\ell+C_r} \Sigma^{r+1}\Sigma^*S (\textit{code})^+
\end{align}
where:
\begin{align*}
	C_i &= \text{value of the $i$-th counter}\\
	\ell &= \text{sufficiently large prefixed number so that interferences with anything}\\
	&\hspace*{2.25em} \text{appearing after~$S$ is not possible in the descriptions below}\\
	\textit{code} &= \text{sequence of states, the first of which corresponds to state of~$A$ being simulated}
\end{align*}

This definition is recursive. The first~$\textit{code}$ after~$S$ denotes the current state corresponding to the state of the counter automaton being simulated. The codes appearing after the first~$\textit{code}$ can be ignored, as they only describe the history of the simulation.

Any serious violations in the above format, due to SSCG rules being employed in an unintended manner, will be caught separately (cf.\ \Cref{sec:formatErrorChecking}).

\medskip

Intuitively, one can think of the counter values and the current state as instantaneous descriptions of the counter automaton, with the string formed by using the characters in~$\Sigma$ in the sentential form as already output by the counter automaton.

\medskip

We now describe how to do the steps in the simulation. To do this, we describe the message generation rules, and show how to implement counter automaton instructions.

\subsection{Rules for message generation}
Below, we give the rules which govern the generation of messages in our SSCG. Messages may be generated at any step of the derivation, though they will only be relevant based on the instruction being processed. If irrelevant messages are generated, then they will be caught in the error checking below (see \Cref{sec:formatErrorChecking}). A message~$M$ is always placed immediately after~$S$ in the sentential form. 

\begin{center}
\setlength{\extrarowheight}{3pt}
\begin{tabular}{rc|c|c}
& Rule & Condition on Subword & Explanation \\
\hline
\rowgap{4}\\
\hline
(M1) & $S \ra SM$ & $T \Sigma^i T^{\ell + 1}$ present & $M$ is ``counter~$i$ is non-zero''\\
(M2) & $S \ra SM$ & $T \Sigma^i T^{\ell + 1}$ absent & $M$ is ``counter~$i$ is zero''\\
(M3) & $S \ra SM$ & $T \Sigma^i UT^{\ell - 1}$ present & $M$ is ``counter~$i$ is being updated''\\
(M4) & $S \ra SM$ & $T \Sigma^i UT^{\ell - 1}$ absent & $M$ is ``counter~$i$ is not being updated''\\
\end{tabular}
\end{center}

We also include the following collection of rules, denoted by~(M5), which handles format and update errors:

\begin{quote}
\hspace*{-1cm}(M5) If~$u$ is any of
\[
	U\Sigma \hspace*{1em} UUTT \hspace*{1em} TU^iT \text{ (for~$0 < i < 4$)} \hspace*{1em} \Sigma UU \hspace*{1em} TTU \hspace*{1em} \Sigma TU \hspace*{1em} \Sigma UTU \hspace*{1em} T \Sigma T \hspace*{1em} U^8
\]
and~$u$ is present/absent, then~$S \ra SM$ where
\[
	M \text{ is ``$u$ is present''/``$u$ is absent''}.
\]
\end{quote}

The role of~(M5) is elucidated when we handle all kinds of sanity checks in \Cref{sec:formatErrorChecking}.

\subsection{Processing of a counter automaton instruction.}

Each counter instruction in~$A$ is processed as follows. We specify the messages which are generated, and the sequence of steps the simulation takes to carry
out the instruction. Further, we explain the process of checking whether the generated messages correspond exactly to the instructions which are being executed.

\medskip

For ease of presentation, we identify codes for (dummy) states/messages with their state/messages.
Suppose the current state is~$s_j$. The state names
\[
	d \, i \, s_j
\]
where~$i \in \mathbb{N}$ are dummy states related to the state~$s_j$. These are required to carry out counter automaton instructions which need several steps in our SSCG.

\medskip

Below, we give the rules of the grammar interspersed with explanations of how the instructions are implemented.

\begin{itemize}
\item[(A)] Checking the counter value of counter~$i$ to be zero/non-zero in any state:
\end{itemize}

\begin{center}
\begin{tabular}{Q{1cm}P{2cm}|P{4.5cm}|P{6cm}}
& Rule & Condition on Subword & Explanation \\
\hline
\rowgap{4}\\
\hline
\multirow{2}*{(1)} & \multirow{2}*{$S \ra Ss_k$} & \multirow{2}*{$SMs_j$ present} & $s_k$ is code for next state\\
&&& $M$ is ``counter~$i$ is zero/non-zero''\\
\end{tabular}
\end{center}

\bigskip

\begin{itemize}
\item[(B)] Printing of a character in a state:
\end{itemize}

\begin{center}
\begin{tabular}{Q{1cm}P{2cm}|P{4.5cm}|P{6cm}}
& Rule & Condition on Subword & Explanation \\
\hline
\rowgap{4}\\
\hline
\multirow{2}*{(2)} & \multirow{2}*{$S \ra aSs_k$} & \multirow{2}*{$Ss_j$ present} & $s_k$ is code for next state\\
&&& $a$ is char being printed
\end{tabular}
\end{center}

\bigskip

\begin{itemize}
\item[(C)] Incrementing/decrementing the $i$-th counter in state~$s_j$:
\end{itemize}

\smallskip

The rules below use the structure of codes as defined in~(\ref{eq:codeform}) in an essential way. Similarly, dummy states, which are needed to simulate the one-step-transition in the counter automaton~$A$, play an important role here. 

\smallskip

\begin{center}
\begin{tabular}{Q{1cm}P{2cm}|P{4.5cm}|P{6cm}}
& Rule & Condition on Subword & Explanation \\
\hline
\rowgap{4}\\
\hline
(3) & $S \ra S d 1 s_j$ & $Ss_j$ present &\\[3mm]
\multirow{2}*{(4)} & \multirow{2}*{$T \ra U$} & \multirow{2}*{$S d 1 s_j$ present} & change first~$T$ in counter\\
&&& which is to be updated\\[3mm]
(5) & $S \ra S d 2 s_j$ & $S d 1 s_j$ present &\\[3mm]
\multirow{4}*{(6)} & \multirow{4}*{$S \ra Sd3s_j$} & \multirow{4}*{$SMd2s_j$ present} & $M$ is sequence of messages:\\
&&& ``counter~$i$ is updated,\\
&&& none other are updated''\\
&&& and no formatting errors\\[3mm]
(7) & $U \ra \emptyword | TT$ & $Sd3s_j$ present & decrement/increment counter\\[3mm]
\multirow{3}*{(8)} & \multirow{3}*{$S \ra Sd4s_j$} & \multirow{3}*{$SMd3s_j$ present} & $M$ is sequence of messages:\\
&&& ``no counter is updated''\\
&&& and no formatting errors\\[3mm]
(9) & $S \ra S s_k$ & $S d 4 s_j$ present & $s_k$ is code for next state
\end{tabular}
\end{center}

\bigskip

\begin{itemize}
\item[(D)] If~$s_j$ is a halting state:
\end{itemize}

\begin{center}
\begin{tabular}{Q{1cm}P{2cm}|P{4.5cm}|P{6cm}}
& Rule & Condition on Subword & Explanation \\
\hline
\rowgap{4}\\
\hline
(10) & $S \ra \emptyword$ & $Ss_j$ present &\\
\end{tabular}
\end{center}

\bigskip

\begin{itemize}
\item[(E)] Termination rules:
\end{itemize}

\begin{center}
\begin{tabular}{Q{1cm}P{2cm}|P{4.5cm}|P{6cm}}
& Rule & Condition on Subword & Explanation \\
\hline
\rowgap{4}\\
\hline
\multirow{2}*{(11)} & $T \ra \emptyword$ & \multirow{2}*{$S$ absent} &\\
& $U \ra \emptyword$ &
\end{tabular}
\end{center}

Note that while~$S$ is present, only the rules in~(C) can update~$U$ and~$T$ in some dummy states. We explain how to check whether that is done consistently in the following section.

\subsection{Checking Format Errors}\label{sec:formatErrorChecking}

\begin{enumerate}[label = (\roman*)]
\item In states where $T$ is converted to $U$:

\begin{itemize}
\item Check that
\[
	U\Sigma \quad \Sigma UU \quad TTU \quad \Sigma TU \quad \Sigma UTU
\]
are not present (this ensures that to the left of~$S$ the only~$T$ converted to~$U$ is the leftmost~$T$ in any counter).
\item Check that~$U^8$ is not present (this ensures that to the right of~$S$ no~$T$ is converted to~$U$). 
  \item Check that there is a~$T\Sigma T$ in the sentential form (this ensures that the first~$T$ to the left of~$S$ is preserved).
 \end{itemize}

\item In states where~$U$ is converted to~$\emptyword$, check that
the codes to the right of $S$ make sense.
No
\[
	UTT \quad TU^iT \text{ (for } 0<i<4 \text{)} \quad T(U^4T)^5 \quad S(U^4T)^5
\]
nor a sequence of six~$TU^{<6}T$ without a~$TU^6T$~/~$TU^7T$ in between, is present.\newline Each adjacent pair of~$TU^6T$, without an~$U^7T$ in between, is separated by five copies of~$TU^{<6}T$.\newline Each pair of adjacent~$TU^7T$ is separated by~$6r'$-many occurrences of~$TU^{<7}T$.

\item When~$U$ is converted to~$TT$, check that there is no~$UTT$.
\end{enumerate}

The above leaves the possibility that the first state codes or messages which appear after~$S$, but before the first state code, are void. However, in that case no progress can be made:~$S$ can never be removed (none of the rules above apply, and any use of reinitialisation below will fail).

\subsection{Initialisation}

We now describe how smaller strings are generated, and how the initial configuration (after printing a string of length $(r+1)(r+2)/2$ by the counter automaton) is obtained.

\begin{center}
\begin{tabular}{Q{1cm}P{2cm}|P{4.5cm}|P{6cm}}
& Rule & Condition on Subword & Comments \\
\hline
\rowgap{4}\\
\hline
\multirow{2}*{(I1)} & \multirow{2}*{$S \ra u$} & \multirow{2}*{$T$ absent} & to generate words~$u$\\
&&& of length~$\le (r+1)(r+2)/2$\\[3mm]
\multirow{3}*{(I2)} & \multirow{3}*{$S \ra v$} & \multirow{3}*{$T$ absent} & $v$ is of the form:\\
&&& $T\Sigma T^{\ell + C_1}T \dots T^{\ell + C_r}\Sigma^{r+1}S C$\\
&&& where $C$ is the start state
\end{tabular}
\end{center}

We use~(I2) to generate strings of length exceeding~$(r+1)(r+2)/2$. Further, we assume in the above that the counter automaton is modified so that, initially, it guesses a string of length~$(r+1)(r+2)/2+1$, and then goes to a particular counter state~$C$ before progressing.

Note that there are always at least two copies of~$T$ to the right of~$S$ when a rule to convert~$T$ to~$U$ is used. Thus, no initialisation rules can be applied after the first initialisation.

Further, if we use a universal counter automaton, then the number of counters is fixed, and we can assign the initial counter values based on the program run by the universal counter automaton. This ensures that in the above construction, the length of the side conditions are in fact independent of the RE language.

\medskip

By the above argument, our SSCG~$G$ simulates the counter automaton~$A$ which we fixed at the beginning of this proof. By the rules of our SSCG, it can be readily verified that~$G$ generates exactly the counter automaton's language~$L$. Hence, the proof is complete. \qedhere

\end{proof}

\begin{rem}[Side condition sizes independent of alphabet size]
In some situations, researchers may wish the condition sizes to be independent of the alphabet size. This only affects the printing rules in the above construction, and, correspondingly, the code size may increase. To handle this, one can proceed as follows.

\medskip

Firstly, assume that the program simulated by a universal counter automaton intends to write the symbol~$a_k$. Depending on this symbol, it guesses whether a pre-determined counter (we use counter~$r+1$ for this purpose) to have value~$k$. in order to verify that the printing was completed successfully, the counter automaton counts down to~$0$ in counter~$r + 1$ using rules~(3) to~(6) above. If the counter didn't have value~$k$, then the automaton never halts. Otherwise, the automaton halts, and it is confirmed that the printing was completed successfully.

\medskip

To simulate the above step, we add the following rules. Below, let~$d$ denote a printing state, and~$d', d''$ dummy printing states.

\begin{center}
\begin{tabular}{Q{1cm}P{2.5cm}|P{4cm}|P{6cm}}
& Rule & Condition on Subword & Comments \\
\hline
\rowgap{4}\\
\hline
(12) & $S \ra a_k T^k S d'$ & $\Sigma S \, d$ present & exists for each~$a_k$ in alphabet\\[3mm]
&\multirow{2}*{$S \ra SM$} & $\Sigma^{r+1}\Sigma T$ present/absent & \multirow{2}*{$M$ codes presence/absence}\\[9mm]
(13) & $S \ra S \, d''$ &$SM \, d'$ present & $M$ codes presence of~$\Sigma^{r+1} \Sigma T$
\end{tabular}
\end{center}

\smallskip

In Rule~(12) above, we assume the alphabet is of the form~$\lbrace a_2,\dots,a_h \rbrace$ for some~$h \in \mathbb{N}$, in order to avoid trivial counter values.

After the above, the universal counter machine proceeds with the simulation, assuming the guess of the counter is complete.

\medskip

Now, the rest of the construction is similar to earlier, except that the~($r+1$)-th counter is considered~$0$ if the sentential form contains a substring of the form~$\Sigma^{r+2} S$. If instead the sentential form contains a substring of the form~$\Sigma^{r+2} T$, then it is non-zero.

\begin{question}
	Can every regular language over a binary alphabet be generated by an SSCG with two non-terminals?
\end{question}

\end{rem}

\bibliographystyle{alpha}
\bibliography{masterbib.bib}

\end{document}